\documentclass[aps,pra,twocolumn,showpacs,showkeys,superscriptaddress,nobalancelastpage,nofootinbib,floatfix,longbibliography]{revtex4-2}

\usepackage{amsmath,amssymb,amsthm}
\usepackage{graphicx}
\usepackage{physics}
\usepackage{enumitem}
\usepackage{xspace}

\usepackage{empheq} % Add this to your preamble

\usepackage{orcidlink}

\usepackage{url}
\usepackage{hyperref}
\usepackage{color}
\definecolor{refcolor}{RGB}{160,35,0}
\definecolor{hrefcolor}{RGB}{0,35,190}
\hypersetup{
    colorlinks,
    citecolor=refcolor,
    filecolor=refcolor,
    linkcolor=hrefcolor,
    urlcolor=hrefcolor
}

\def\({\left(}
\def\){\right)}
\def\[{\left[}
\def\]{\right]}

\newcommand{\hilbert}{\mathcal{H}}
\newcommand{\pobs}[1]{#1}
\newcommand{\obs}[1]{\mathsf{\pobs{#1}}}

\newcommand{\mc}[1]{\mathcal{#1}}

\newcommand{\R}{\mathbb{R}}

\newcommand{\WdWSpaceL}{{\mkern-8mu}}
\newcommand{\WdWSpaceM}{{\mkern-6mu}}
\newcommand{\WdWSpaceMM}{{\mkern-9mu}}
\newcommand{\WdWSpaceR}{{\mkern-8mu}}
\newcommand{\braa}[1]{\bra{\WdWSpaceL\bra{#1}\WdWSpaceM}}
\newcommand{\kett}[1]{\ket{\WdWSpaceM\ket{#1}\WdWSpaceR}}

\newcommand{\brakett}[2]{\left.\bra{#1}{\WdWSpaceMM}\ket{#2}{\WdWSpaceR}\right\rangle}
\newcommand{\evtt}[2]{\braa{#2}{\mkern-3mu}{#1}{\mkern-4mu}\kett{#2}}

\newcommand{\q}{\boldsymbol{q}}
\newcommand{\M}{\mathbf{M}}

\newcommand{\vs}{\textit{vs.}\xspace}

\newtheorem{definition}{Definition}
\newtheorem{proposition}{Proposition}
\newtheorem{lemma}{Lemma}

\newtheorem{principle}{Principle}

\newtheorem{implication}{Implication}

\theoremstyle{remark}

\newtheoremstyle{nospace} % name of the style
  {3pt}                   % Space above
  {3pt}                   % Space below
  {\itshape}              % Body font
  {}                      % Indent amount
  {\bfseries}             % Theorem head font
  {.}                     % Punctuation after theorem head
  {.5em}                  % Space after theorem head
  {\thmname{#1}\thmnumber{#2}\thmnote{ (#3)}} % <--- HEADER SPECIFICATION

\theoremstyle{nospace}

\newtheorem{defCustom}{}
\makeatletter
\newcommand{\setdefCustomtag}[1]{% \setdefCustomtag{<tag>}
  \let\oldthedefCustom\thedefCustom% Store \thedefCustom
  \renewcommand{\thedefCustom}{{\normalfont\textbf{#1}}}% Redefine it to a fixed value
  \g@addto@macro\enddefCustom{% At \end{defCustom}, ...
    \global\let\thedefCustom\oldthedefCustom}% ...restore \thetheorem
  }
\makeatother

\allowdisplaybreaks

\begin{document}

\title{Exact quantum time compatible with positive energy}

\author{\orcidlink{0000-0002-2765-1562}\ O.C. Stoica}
%\email[]{cristi.stoica@theory.nipne.ro,holotronix@gmail.com}
\affiliation{Dept. Th. Physics, NIPNE-HH, Bucharest, Romania. \href{mailto:cristi.stoica@theory.nipne.ro}{cristi.stoica@theory.nipne.ro},  \href{mailto:holotronix@gmail.com}{holotronix@gmail.com}}

\keywords{Quantum time; time operator; Wheeler-DeWitt equation; Hamiltonian bounded from below; Unruh-Wald theorem; Hegerfeldt-Ruijsenaars theorem.}

\begin{abstract}
What would it be like to be in a superposition of yesterday, today, and tomorrow? This question may seem at best entertaining, but it is necessary, and exploring it allows us to understand how exact irreversible clocks and change are possible, despite the Unruh-Wald and Hegerfeldt-Ruijsenaars no-go theorems forbidding them.

Unruh and Wald (1989) proved that if energy is bounded from below, no observable can increase monotonically with the Schr\"odinger time parameter $t$. Perfectly monotonic clocks and irreversible observable changes (Hegerfeldt-Ruijsenaars, 1980) seem impossible. From the perspective of the Schr\"odinger time, the world appears in a superposition of different intrinsic clock states indicating different times and opposite time directions. This seems to directly contradict our daily experiences of time and change.

I show that there is no contradiction: from an intrinsic perspective of the world, sharp irreversible changes do happen, because the macroscopic pointer states resolve the superposition of different times. Large-scale time-reversing or discontinuous transitions are not internally observable in the records. From the intrinsic perspective, an unbounded intrinsic-time translation generator plays the role of the Hamiltonian, generating only forward time evolution with respect to the intrinsic time, but not to the Schr\"odinger parameter $t$, which is thus not justified to play the role of time. This allows sharp time observables even if the external Hamiltonian is bounded from below. In addition, this leads to a stationary wavefunction of the universe satisfying a Wheeler-DeWitt-type equation, without assuming gravity.
\end{abstract}

\maketitle

\begin{quote}
\emph{Or how could it come into being? If it came into being, it is not; nor is it if it is going to be in the future. Thus is becoming extinguished and passing away not to be heard of nor is it divisible, since it is all alike, and there is no more of it in one place than in another, to hinder it from holding together, nor less of it, but everything is full of what is.}
\end{quote}
\begin{flushright}
Parmenides, \emph{Fragments}.
\end{flushright}

%------------------------------------------------------------%
\section{Introduction}
\label{s:intro}

In ordinary Schr\"odinger quantum mechanics, time is not an observable like position, momentum, or spin, it is a parameter with respect to which the state evolves:
\begin{equation}
\label{eq:time-evol}
\ket{\psi(t)}=e^{-\frac{i}{\hbar}\obs{H}t}\ket{\psi(0)}.
\end{equation}

The state vector $\ket{\psi(t)}$ depends on time, but $t$ is not represented by a self-adjoint operator on the Hilbert space of the system.
The hope to represent it by a self-adjoint operator was deflated by Pauli's theorem~\cite{Pauli1980GeneralPrinciplesOfQuantumMechanics}, showing that if there were a self-adjoint operator $\obs{T}$ satisfying
\begin{equation}
\label{eq:T-H-conjugate}
[\obs{T},\obs{H}]=i\hbar,
\end{equation}
then the Hamiltonian $H$ could not be bounded below.

In a seminal article containing several other important results~\cite{UnruhWald1989TimeAndTheInterpretationOfCanonicalQuantumGravity}, Unruh and Wald strengthened Pauli's result to any possible observable $\obs{T}$, not assumed to be canonically conjugate with the Hamiltonian.
They only assumed that Hamiltonian is bounded from below,
\begin{equation}
\label{eq:H-bounded-below}
\obs{H}\geq \mathfrak{c}
\end{equation}
for some $\mathfrak{c}\in\R$, and that the time observable $\obs{T}$ correlates monotonically with the Schr\"odinger parameter $t$.
Consider two clock states $\ket{\tau_n}$ and $\ket{\tau_m}$ with $m>n$, assuming that $\ket{\tau_n}$ precedes $\ket{\tau_m}$. Unruh and Wald explored what happens if once the system's state becomes $\ket{\tau_m}$, it never returns to $\ket{\tau_n}$ according to the Schr\"odinger time. This condition, called the \emph{Heraclitean property} of time, is that the amplitude for a clock to go backward from the later reading $\tau_m$ to the earlier reading $\tau_n$ after a positive Schr\"odinger time $t$
\begin{equation}
f(t)=\bra{\tau_n}e^{-i \obs{H} t/\hbar}\ket{\tau_m}
\end{equation}
should vanish for $t>0$.
Unruh and Wald showed that the conditions $f(t)=0$ for $t>0$ and $\obs{H}\geq \mathfrak{c}$ together imply that $f(t)=0$ for all real $t$.

The meaning of this result cannot be overstated. If there were a time observable $\obs{T}$ able to indicate that time passed, it is reasonable to ask that it does not go back to indicating a past time. But this is impossible if the Hamiltonian is bounded from below.

Their result is perfectly rigorous, and can be easily derived from a Lemma proved earlier by Hegerfeldt and Ruijsenaars~(\cite{HegerfeldtRuijsenaars1980RemarksOnCausalityLocalizationAndSpreadingOfWavePackets}, p. 378):
\begin{lemma}[Hegerfeldt \& Ruijsenaars]
\label{lemma:HegerfeldtRuijsenaars}
Let $\obs{H}$ be a self-adjoint operator on a Hilbert space $\hilbert$, $\hilbert'$ a closed subspace of $\hilbert$, and $\ket{\psi_0}\in\hilbert$. If $\obs{H}$ is bounded from below and $e^{-\frac{i}{\hbar}\obs{H}t}\ket{\psi_0}\in\hilbert'$ for all values of $t$ in an open interval $(t_1,t_2)$, then $e^{-\frac{i}{\hbar}\obs{H}t}\ket{\psi_0}\in\hilbert'$ for all $t\in\R$.
\qed
\end{lemma}

This Lemma is often stated as \emph{``nothing can happen for the first time''}. Mathematically, both results can be proved by using the analytic extension of the function $f(t)$ in the upper complex plane, for example~\cite{StreaterWightman1964PCTSpinStatisticsAndAllThat}. It was used by Hegerfeldt in~\cite{Hegerfeldt1974RemarkOnCausalityAndParticleLocalization}, another gem that deserves more attention.
The Unruh-Wald theorem can be proved easily by applying the Lemma~\ref{lemma:HegerfeldtRuijsenaars} to $\hilbert'=\ket{\tau_n}^\perp:=\{\ket{v}\in\hilbert|\braket{v}{\tau_n}=0\}$ and $(t_1,t_2)=(0,\infty)$.

Moreover, while Unruh and Wald did not assume condition~\eqref{eq:T-H-conjugate}, the condition that no (generalized) eigenstate $\ket{\tau}$ of $\obs{T}$ can evolve to a past eigenstate of $\obs{T}$ can be satisfied only if the Hamiltonian satisfies condition~\eqref{eq:T-H-conjugate}~\cite{Stoica2022ProblemOfIrreversibleChangeInQuantumMechanics,Stoica2024EmpiricalAdequacyOfTimeOperatorCC2HamiltonianGeneratingTranslations}.
In this case, the Hamiltonian has, in some parametrization $\tau\in\R$, the form
\begin{equation}
\label{eq:clock-H-irrev}
\obs{H}=-i\hbar\pdv{\tau}.
\end{equation}

This seems too restrictive, but in fact there are many plausible systems whose Hamiltonian can be expressed like this, for example systems that contain an ideal clock~\cite{Page1986DensityMatrixOfTheUniverse,Stoica2026TheClockAmbiguityProblemExtendedOrExtinguished}, systems that contain an ideal measurement~\cite{BuschGrabowskiLahti1995OperationalQuantumPhysics,Mittelstaedt2004InterpretationOfQMAndMeasurementProcess}, and Koopman-von Neumann quantum representations of irreversible classical systems~\cite{Koopman1931HamiltonianSystemsAndTransformationInHilbertSpace,vonNeumann1932KoopmanMethod}, as seen in~\cite{Stoica2022ProblemOfIrreversibleChangeInQuantumMechanics,Stoica2024EmpiricalAdequacyOfTimeOperatorCC2HamiltonianGeneratingTranslations}.

To avoid the problems raised by the Pauli and Unruh-Wald theorems, a usual strategy is to consider unsharp clock observables (positive operator-valued measure--POVM)~\cite{Busch2002Time-energy-uncertainty-relation,LoveridgeMiyadera2019RelativeQuantumTime,HohnSmithLock2021TrinityOfRelationalQuantumDynamics,GemsheimRost2023EmergenceOfTimeFromQuantumInteractionWithTheEnvironment}.
This may seem the only option left by the Unruh-Wald theorem, so maybe we are doomed to never afford a good enough clock.
But the problem does not just go away by invoking unsharp observables, because this allows a nonzero amplitude for the system to transition to a past state or jump to a distant future state.

Moreover, the problem is much broader, because the Hegerfeldt-Ruijsenaars shows more than this, it shows that no definite observable change is possible: if a property has some value for a while (enough to be observed), it will have the same value all the time.
If its value is bounded for a while in some interval, it will be bounded to that interval forever. If the observable property changes its value continuously, its value is bounded to any interval no matter small, so it can never change.

In particular, we would not be able to prepare or measure states with perfect accuracy, so we would not be able to accurately exchange information~\cite{Stoica2025CanWeAccuratelyReadWriteQuantumData}.
Admitting that everything we know about the world is simply unsharp~\cite{BuschJaeger2010UnsharpQuantumReality} means giving up any hope of accurate knowledge of the world.

Here I will accept the result and embrace it, focusing on a different lesson that can be understood from the Hegerfeldt-Ruijsenaars Lemma and Unruh-Wald theorem: maybe we should give up the Schr\"odinger parameter $t$ as time in the usual sense, as that which accounts for change.
Change is recorded through observables, so temporal ordering should be reconstructed from records, and not identified \emph{a priori} with the Schr\"odinger parameter $t$, even if this leads us to abandon the Heraclitean property of time.

I will consider candidate time observables $\obs{T}$ acting on the Hilbert space $\hilbert$ of a large closed system, which could be the universe.
The observable $\obs{T}$ should not be conserved, because then its value will not change.
Since there are many other observables on $\hilbert$ that can be measured together with $\obs{T}$, its spectrum has to be highly degenerate and uniform.
\begin{definition}
\label{def:reaches}
If a state $\ket{\psi}$ is a generalized eigenstate of $\obs{T}$ with eigenvalue $\tau$, we say that $\ket{\psi}$ has definite intrinsic time $\tau$. 
If $\ket{\psi}$ has a nonzero component in any open interval of the spectrum of $\obs{T}$ containing $\tau$, we say that it \emph{reaches the intrinsic time} $\tau$.
\end{definition}

Lemma~\ref{lemma:HegerfeldtRuijsenaars} implies the following
\begin{implication}
\label{implication:UW}
If the system has a definite intrinsic time $\tau_0$ at the Schr\"odinger time $t_0$, it evolves immediately into a superposition of states with different intrinsic times.
Moreover, during any open Schr\"odinger time interval $(t_1,t_2)$, the system reaches all intrinsic times it reaches for any other Schr\"odinger time, past or future.
\end{implication}
\begin{proof}[Explanation]
From Lemma~\ref{lemma:HegerfeldtRuijsenaars}, any eigenstate of $\obs{T}$ orthogonal to the state of the system for an open Schr\"odinger time interval is orthogonal to the system's state at all times.
This implies that if the system reaches at least once (with respect to the Schr\"odinger time) an intrinsic time $\tau$, it reaches it in any open Schr\"odinger time interval, no matter how small.
\end{proof}

In Section~\ref{s:superposition-of-times} we will explore how this strange situation is perceived from within a large system that observes change, tracking it as records using the observable $\obs{T}$.
We will see that a generic ingredient from Everett's answer to the measurement problem and the problem of emergence of classicality~\cite{Everett1957RelativeStateFormulationOfQuantumMechanics,Everett1973TheTheoryOfTheUniversalWaveFunction}, at least in its modern understanding as in~\cite{Wallace2012TheEmergentMultiverseQuantumTheoryEverettInterpretation,SEP-Vaidman2021MWI}, gives an answer to this question as well, with minimal adaptations.
The present $\tau$-time as experienced by an observer is reached at all Schr\"odinger times. This uncanny situation is a direct consequence of the Lemma~\ref{lemma:HegerfeldtRuijsenaars} by Hegerfeldt-Ruijsenaars.
But, for an observer within the world, this is not a paradoxical situation, since the $\tau$-time works as one would expect from a ``legitimate'' notion of time.

In Section~\ref{s:intrinsic-time-wavefunction} we explore quantum time observables and how they are defined as collective records in pointer states. Rejecting the Schr\"odinger parameter $t$ as playing the role of time, we find that the Wheeler-DeWitt equation holds even in quantum theories that do not contain gravity implicitly. We derive the effective intrinsic Schr\"odinger-like equation and study its main properties.

%------------------------------------------------------------%
\section{Superposition of times}
\label{s:superposition-of-times}

%------------------------------------------------------------%
\subsection{What it is like to be in a superposition of yesterday, today, and tomorrow?}
\label{s:yesterday-today-tomorrow}

It is Friday. Let $\ket{\psi_{T}}$, $\ket{\psi_{F}}$, and $\ket{\psi_{S}}$ represent three quantum states of the world, respectively the world as it was yesterday this time, as it is now, and as it will be tomorrow same time.
Since for the internal observers there are clear differences between the state of the world at these three times, these three vectors span a three-dimensional vector space. Are all the other unit vectors from this space, except for these ones, forbidden by Nature? There is no reason to believe this, quite the opposite, states like
\begin{equation}
\label{eq:superposition-times}
\ket{\psi(t_{F})}=\alpha\ket{\psi_{T}}+\beta\ket{\psi_{F}}+\gamma\ket{\psi_{S}}
\end{equation}
are much more generic than each of the three states.
There are infinitely many of them compared to only three.

One may object that this is impossible because it is excluded by our observations. We never observe a superposition of the world at different times as in~\eqref{eq:superposition-times}.
Yesterday is already gone and tomorrow did not arrive yet, so how could such a state be possible?
But even if we observe that the world is in the state $\ket{\psi_{F}}$, this does not show that it cannot be in a superposition of the form~\eqref{eq:superposition-times}, with $\beta\neq 0$.
An observer that is part of $\ket{\psi_{F}}$ would only see the world consistent with being in the state $\ket{\psi_{F}}$, and not $\ket{\psi_{T}}$ or $\ket{\psi_{S}}$, even if these may be equally real, containing a copy of the observer thinking that the actual state of the world is in fact $\ket{\psi_{T}}$ or $\ket{\psi_{S}}$.
All subsystems of the state $\ket{\psi_{F}}$ have mutually conditioned relative states, and they are oblivious to the mutually conditioned relative states of the subsystems of $\ket{\psi_{T}}$ and $\ket{\psi_{S}}$.

%------------------------------------------------------------%
\subsection{Schr\"odinger's cat and the primordial soup}
\label{s:schrodinger-cats}

I hope that at least Everettians accept the physical possibility of states like~\eqref{eq:superposition-times}, but if this still seems too implausible, despite their unavoidability, let us revisit the Schr\"odinger cat experiment~\cite{schrodinger1935SchrodingerCat}.
To avoid paperwork, there will be a box but no cat. There will be a radioactive substance and the Geiger counter that triggers a mechanism that smashes with a hammer a flask, but instead of hydrocyanic acid we will use a colored harmless gas.
As the usual story goes, after a time $T$ the box will be in a superposition
\begin{equation}
\label{eq:schrodinger-box}
e^{-\frac{\Gamma T}{2}}\ket{\text{no gas released}}+\sqrt{1-e^{-\Gamma T}}\ket{\text{gas released}}
\end{equation}

However, the released gas fills the box gradually, not instantaneously, so there is a continuum of possible states in which the gas was released, depending on the time when this happened. At time $T$, there is in fact a  superposition of many states in which the gas was released at different times between $0$ and $T$.
So there will be a state of the box in which the gas was not released yet, one in which it was just released at $t=T$, one in which it was released a microsecond before and so on, up to a state in which the gas was released at $t=0$.

But this is not so different from the box being in a superposition of its states at different times, of different stages of the gas filling the box.

One may imagine similar machineries at a cosmic scale, for example a universe in which the primordial false vacuum could decay into an inflationary universe at different times, resulting in a superposition of universes that started to expand at different times.

But such contraptions are needed only for our imaginations, we do not need to invoke them as a mechanism that leads the world in a superposition of different times.
All we needed this thought experiment for was to accommodate ourselves with the possibility.
Nothing forbids a universe that started in a state that then evolved into~\eqref{eq:superposition-times}, such as $e^{\frac{i}{\hbar}\obs{H} T_{u}}\(\alpha\ket{\psi_{T}}+\beta\ket{\psi_{F}}+\gamma\ket{\psi_{S}}\)$, where $T_{u}$ is the age of the universe, or even the more general $\sum_j\alpha_j e^{\frac{i}{\hbar}\obs{H} (T_{u}-\theta_j)}\ket{\psi_{F}}$.
Not only this, but it seems more difficult to constrain somehow the universe to \emph{not} be in a superposition of different intrinsic times.
Quite the opposite, forcing an initial condition so that today the state is $\ket{\psi_{F}}$ requires a very strong fine tuning of the initial conditions, and this will work only for an instant, since the state $\ket{\psi_{F}}$ will immediately reach the states from yesterday, today, tomorrow, and any other day in the past and future history.

%------------------------------------------------------------%
\subsection{Relative states distinguished by time records}
\label{s:decoherence-time}

Applying the same logic as in Everett's relative states, the reason why we do not see the world in a superposition between yesterday and today and tomorrow is that there is a copy of us within the state $\ket{\psi_{T}}$, who thinks it is Thursday, and also the calendars, coworkers, family, friends, newspapers, all confirm it is Thursday, and similarly for the states $\ket{\psi_{F}}$ and $\ket{\psi_{S}}$.
It is similar to the social media bubbles where everyone agrees on certain topics, but as soon as you talk with people from other bubbles, they seem to live in a completely different reality.
Except that in this case, you cannot leave your bubble to meet someone, including yourself or your grandfather, who lives in a different intrinsic time. Not because this would be tantamount to time travel, neither due to decoherence, which prevents jumps from an Everettian branch of the wavefunction to another branch, but because different times are distinguished by different records, and the clocks agree at the level of each branch.
Each of the states in the superposition is a relative state consisting of subsystems corroborating the same information about the time.
The Unruh-Wald theorem implies that the wavefunction is a superposition of states carrying both earlier and later intrinsic records. It seems to jump back and forth in time, but the moments of time are defined by the pointer states, the records, and therefore these do not jump.
These records allows different time branches be decohered and exhibit a degree of stability that allows us to trust them as recording past events.

There is a useful way to understand Everett's interpretation that may help in superposition of different intrinsic times.
When we do a quantum measurement, we never see the actual state of the microscopic system we measure, all we can is to read the state of the pointer, which is macroscopic.
Pointer states can be ``copied'': they can be redundantly and robustly correlated with other degrees of freedom of the environment. For this reason we can read them without disturbing them as when we do quantum measurements of microscopic properties, because we can access macroscopic observables with our senses.

If the pointer is a needle that moves to right or left, depending on the result, we observe its position. The position of the pointer is a coarse graining, an aggregate function of the observables of the particles constituting the pointer. As an observable, it applies to the whole pointer, not to the individual particles, and it is redundantly encoded in the state of the pointer, because we can read it even if we see it only partially.
Distinguishable pointer states correspond to orthogonal sectors of the Hilbert space, and thus they are decohered.

Our phenomenal world is constituted of such pointer observables.
But pointer observables are quantum observables too, the difference is that they apply to system composed of a larger number of particles, and are relatively much more stable.

These macroscopic observables correlate with each other, for example, if Alice does a spin measurement and obtains that the spin is up, she can communicate this to Bob, and Bob can spread the word in the entire University. In the parallel branch, where Alice obtained that the spin is down, the word will spread that the spin is down. Moreover, there will be records in the machine that show what result was obtained, and all of these can be used to tell apart the two branches of the wavefunction.

This is the old Everettian way to distinguish among branches, there is nothing special about the records indicating the current clock time compared to other pointer observables. It is not even strictly Everettian, as it is employed by all approaches to quantum measurements and macroscopic classicality, since they all rely eventually on decoherence, but the Everettian language comes handy here.

The only difference between clock states and generic pointer states is that we are so used to think that the Schr\"odinger parameter $t$ must be the actual time, and our clocks should just try to emulate it as closely as possible.

One may object that there is no time observable in quantum theory, time is a parameter, while pointer observables are dynamical variables. But note that the time parameter $t_{F}$ appearing in $\ket{\psi(t_{F})}$ from equation~\eqref{eq:superposition-times} is not the same as the pointer observables that the observer consults to know the time. Rather, the three states $\ket{\psi_{T}}$, $\ket{\psi_{F}}$, and $\ket{\psi_{S}}$ are eigenstates of $\obs{T}$ for different eigenvalues $\tau_{T}$, $\tau_{F}$, and $\tau_{S}$.
The clock states are just like all other macroscopic or pointer observables that distinguish between macroscopically distinct states.

%------------------------------------------------------------%
\subsection{Schr\"odinger's time parameter \vs intrinsic time}
\label{s:schrodinger-time-vs-intrinsic-time}

Implication~\ref{implication:UW} of Lemma~\ref{lemma:HegerfeldtRuijsenaars} says that, from the perspective of Schr\"odinger's time parameter, the intrinsic time is smeared along the entire Schr\"odinger time axis.
On the other hand, the intrinsic time is based on conventions, on choosing some macroscopic observables as a measure of change. But one can properly translate between different conventions, so it makes sense to consider intrinsic time as objective. And from the perspective of the intrinsic time, the Schr\"odinger time is not properly defined. One may hope that they are strongly correlated, even if they are not monotonically correlated in the sense considered by Unruh and Wald~\cite{UnruhWald1989TimeAndTheInterpretationOfCanonicalQuantumGravity}. But Implication~\ref{implication:UW} already showed that this is not the case: any Schr\"odinger time interval, no matter how small, reaches all intrinsic times.

%------------------------------------------------------------%
\section{The intrinsic time wavefunction}
\label{s:intrinsic-time-wavefunction}

%------------------------------------------------------------%
\subsection{Macrostates and observable changes}
\label{s:macrostates-changes}

Macroscopic observables are in general associated to composite systems, and can be expressed as coarse grainings, which are aggregate functions of the observables of their constituents. Ideally, they commute with one another. I assume that we can choose a set of commuting macroscopic observables $\widehat{\M}=(\obs{M}_1,\obs{M}_2,\ldots)$ so that all collective macroscopic observables can be obtained as functions of $\widehat{\M}$.

Let $\mc{A}(\widehat{\M})$ denote the resulting algebra of macroscopic observables.
Whatever observables we interpret as monitoring change, in particular our time observable $\obs{T}$, are among the pointer observables from the algebra $\mc{A}(\widehat{\M})$.

An additional set of observables is needed to obtain a complete set of commuting observables, and we will denote it by $\widehat{\q}=(\widehat{q}_1,\widehat{q}_2,\ldots)$.
Then, we can express a vector $\kett{\Psi}$ in a common eigenbasis of $\obs{T}$, $\widehat{\M}$, and $\widehat{\q}$,
\begin{equation}
\label{eq:tau-wavefunction}
\Psi(\tau,\M,\q)=\brakett{\tau,\M,\q}{\Psi}.
\end{equation}

Note that, since $\obs{T}\in\mc{A}(\widehat{\M})$, the configuration space of $\Psi(\tau,\M,\q)$ is a bit redundant. 
We can avoid this by choosing $\widehat{\M}$ to contain only macroscopic observables that are independent of the clock observable $\obs{T}$, but truth is that the clock observables are not independent from the other macroscopic observables. For example, the hands of a clock have definite positions, and we interpret these positions as indicating time.
That is, $\tau=\tau(\M)$, so $\Psi(\tau,\M,\q)=\Psi(\tau(\M),\M,\q)$.
But for now, we can ignore this redundancy even if $\tau$ may have an implicit value or relative value for each macrostate, and use equation~\eqref{eq:tau-wavefunction} to make $\tau$ explicit.

%------------------------------------------------------------%
\subsection{From Schr\"odinger to Wheeler-DeWitt}
\label{s:schrod-t-WdW}

In equation~\eqref{eq:tau-wavefunction}, I denoted the total vector by $\kett{\Psi}$ rather than $\ket{\Psi}$ because it plays a similar role to the stationary solution of the Wheeler-DeWitt-type equation, as we shall see.

On the other hand, the wavefunction from the Schr\"odinger equation is of the form
\begin{equation}
\label{eq:tau-t-wavefunction}
\psi(\tau,\M,\q;t)=\braket{\tau,\M,\q}{\psi(t)},
\end{equation}
depending on the Schr\"odinger parameter $t$.
The Schr\"odinger time evolution equation is
\begin{equation}
\label{eq:time-evol-tau-t}
i\hbar\pdv{t}\psi(\tau,\M,\q;t)=\matrixel{\tau,\M,\q}{\obs{H}}{\psi(t)}.
\end{equation}

There is no way to access the Schr\"odinger parameter $t$ from within the system, so we need to rely only on the observables $\tau$, $\M$, $\q$, and other observables derived from these and their eventual canonical conjugates. We may think that we can access the Schr\"odinger parameter $t$, because not only Schr\"odinger came up with his equation in the first place, but its predictions were confirmed numerous times for more than a century. But note that we never read directly the Schr\"odinger time, rather, we always read a clock intrinsic to our world, but extrinsic to the observed system. The changes of the clock system correlate with the changes of the observed system \emph{as subsystems of our world}.

Since one can never access the Schr\"odinger parameter $t$, the assumption that it plays the role of time, and especially that there is really a unique state vector that evolves in time, becomes unjustified.
The actual time observables are operators like $\obs{T}$, and therefore intrinsic time is parametrized by its spectrum values $\tau$.
This is consistent with the fact that the configuration space is parametrized only by $(\tau,\M,\q)$.

If one adopts the eternalist position of spacetime as a block universe, from relativity, from the intrinsic point of view, the wavefunction does not change with the Schr\"odinger parameter $t$. The actual, timeless wavefunction is the accumulation of the Schr\"odinger wavefunction for all values of $t$, and the Schr\"odinger equation simply describes how various components of the timeless wavefunction connect with each other.
What follows does not necessarily commit to the eternalist view. Even if one prefers the presentist view of spacetime, according to which only the present state exists, since the Schr\"odinger parameter $t$ does not correspond to the present state of the universe, one can consider the present state as the $\tau=\text{const.}$ section of the accumulated wavefunction.

Let us encode the above rationale in a Principle:
\begin{principle}
\label{pp:timeless-wavefunction}
The Schr\"odinger parameter $t$ is not time. The states $\psi(t)$ equally exist for all values of $t$, composing an accumulated timeless wavefunction $\Psi$.
\end{principle}

The macroscopic observables $\obs{T}$ and $\widehat{\M}$ decompose the Hilbert space in sectors that we will call \emph{macrostates}. Two states from the same macrostate are macroscopically indistinguishable. Two states from distinct macrostates can be distinguished macroscopically, and decompose $\ket{\psi(t)}$ into branches.
We will not discuss the Everettian branches resulting from distinct values of the macroscopic observables $\widehat{\M}$ but equal values of $\tau$, they are the subject of the many-worlds interpretation.
Here we will focus on the cases when the differences are made by an observable $\obs{T}$, decomposing $\ket{\psi(t)}$ into states of distinct intrinsic eigenvalues $\tau$ of $\obs{T}$.
Since the same $\tau$ is reached in any open Schr\"odinger time interval, we will take the group-average over the $t$-translations,
\begin{equation}
\label{eq:tau-wavefunction-from-t}
\Psi(\tau,\M,\q)=\int_{\R}\psi(\tau,\M,\q;t)\dd t.
\end{equation}

The timeless vector $\kett{\Psi}$ from equation~\eqref{eq:tau-wavefunction} is simply
\begin{equation}
\label{eq:tau-state-from-t}
\kett{\Psi}=\int_{\R}\ket{\psi(t)} \dd t.
\end{equation}

This can be done more rigorously using the rigged Hilbert space or the refined algebraic quantization framework, but we will not repeat this construction here as it was done numerous times in the literature for various other applications~\cite{Marolf2002GroupAveragingAndRefinedAlgebraicQuantizationWhereAreWeNow,HohnSmithLock2021TrinityOfRelationalQuantumDynamics}.

We avoid normalizing the vector $\kett{\Psi}$ unless its norm is finite, which is not true in general, since the amplitude of $\kett{\Psi}$ spreads across a possibly infinite range of intrinsic times $\tau$ and group-averaging is in general distributional, necessitating the rigged Hilbert space.
This situation resembles that of the Page-Wootters proposal~\cite{Page1986DensityMatrixOfTheUniverse}, and like in that case, it allows us to define probabilities for the other observables $\widehat{\M}$ and $\widehat{\q}$ conditioned by the intrinsic time $\tau$.

What Hamiltonians and choices of the intrinsic time observable yield good results for the so far formal definition of $\kett{\Psi}$ and $\Psi(\tau,\M,\q)$ from equations~\eqref{eq:tau-wavefunction-from-t} and~\eqref{eq:tau-state-from-t}, allowing them to converge?

For a generalized energy eigenvector $\ket{E}$, the integral $\int_{\R} e^{-\frac{i}{\hbar}\obs{H}t}\ket{E} \dd t=\int_{\R} e^{-\frac{i}{\hbar}E t}\ket{E} \dd t=2\pi\hbar\delta(E)$, so it vanishes unless $E=0$. More precisely, the group average is proportional to the operator-valued distribution delta of the Hamiltonian.
Then, basically, the integral $\int_{\R}\ket{\psi(t)} \dd t$ filters out all energy components except for the zero-energy component of $\ket{\psi(0)}$.

Therefore, to get $\kett{\Psi}\neq 0$ in equation~\eqref{eq:tau-state-from-t}, it is necessary that $0$ is in the spectrum of $\obs{H}$ and the zero-energy component of $\ket{\psi(0)}$ is not zero.
Then, we proved that
\begin{proposition}
\label{thm:WDW}
The Wheeler-DeWitt constraint equation holds for $\kett{\Psi}$ from equation~\eqref{eq:tau-state-from-t},
\begin{equation}
\label{eq:WDW}
\obs{H}\kett{\Psi}=0.
\end{equation}
\end{proposition}

Note that we obtained it without assuming that $\obs{H}$ includes the gravitational interactions, unlike the standard situations where the Wheeler-DeWitt equation is obtained~\cite{Dewitt1967QuantumTheoryOfGravityI_TheCanonicalTheory,Dewitt1967QuantumTheoryOfGravityII_TheManifestlyCovariantTheory}.

Equation~\eqref{eq:WDW} says that $\kett{\Psi}$ is a stationary state, a zero-energy (generalized) eigenvector. Therefore, a minimal necessary condition for the integral from equation~\eqref{eq:tau-state-from-t} is that $0$ is contained in the relevant spectral support in $\sigma\(\obs{H}\)$, the spectrum of $\obs{H}$.
If this is not the case, the regularized integral converges to $\kett{\Psi}=0$, and both these situations are incompatible with the mechanism proposed here.

If $0\notin\sigma\(\obs{H}\)$, we can choose an eigenvalue $E_0\in\sigma\(\obs{H}\)$ and shift $\obs{H}$, replacing it by $\obs{H}-E_0\obs{I}$.
Normally, this does not change the physics described by the ordinary Schr\"odinger equation, since it amounts to multiplying the state vector $\ket{\psi(t)}$ by a global phase factor. However, this global phase makes a difference for the cumulative state vector $\kett{\Psi}$ defined by the integral~\eqref{eq:tau-state-from-t}, and for the $\tau$-wavefunction $\Psi(\tau,\M,\q)$ defined in~\eqref{eq:tau-wavefunction-from-t}, so we have to use it with caution. If we take the $\tau$-physics as the real internal perspective of the observers embedded in the world, the difference is important, even if in the $t$-physics there is no difference.

%------------------------------------------------------------%
\subsection{Intrinsic time observables}
\label{s:intrinsic-time-observable}

For a good intrinsic time observable $\obs{T}$, we assume a uniform spectral multiplicity, so that the states at different intrinsic times live in equivalent Hilbert spaces.

It is also possible that the intrinsic time observable $\obs{T}$ is chosen to give a uniform time. If its spectrum is discrete, we can take, for all $\tau_1,\tau_2$,
\begin{equation}
\label{eq:tau-uniform-discrete}
\evtt{\obs{P}_{\tau_1}}{\Psi}=\evtt{\obs{P}_{\tau_2}}{\Psi}=1.
\end{equation}

If the spectrum of $\obs{T}$ is continuous, for a Borel interval $\mathcal{T}\subseteq\R$ in the spectrum of $\obs{T}$, we denote by $\obs{P}_{\mathcal{T}}$ the projection operator on the corresponding subspace $\hilbert_{\mathcal{T}}<\hilbert$ of $\obs{T}$.
Then, we can require
\begin{equation}
\label{eq:tau-uniform-continuum}
\evtt{\obs{P}_{\mathcal{T}_1}}{\Psi}=\evtt{\obs{P}_{\mathcal{T}_2}}{\Psi}<\infty
\end{equation}
for all finite length Borel intervals $\mathcal{T}_1$ and $\mathcal{T}_2$ related by a spectral translation.
We can impose the condition

\begin{equation}
\label{eq:tau-basis}
\braket{\psi(\tau)}{\psi(\tau')}_c=\delta(\tau-\tau')
\end{equation}
 for all $\tau,\tau'$, where, in the rigged Hilbert space,
\begin{equation}
\label{eq:psi-tau}
\ket{\psi(\tau)}:=\lim_{\Delta\tau\to 0}\frac{\obs{P}_{[\tau,\tau+\Delta\tau]}\kett{\Psi}}{\sqrt{\evtt{\obs{P}_{[\tau,\tau+\Delta\tau]}}{\Psi}}}.
\end{equation}

The condition~\eqref{eq:tau-basis} ensures uniformity with respect to $\tau$. Non-uniform reparametrizations of the spectrum preserve orthogonality, but break the condition~\eqref{eq:tau-basis} due to the scaling property of the Dirac distribution.

\begin{definition}
\label{def:tau-uniform}
An intrinsic time observable $\obs{T}$ is called \emph{$\tau$-uniform} if it satisfies condition~\eqref{eq:tau-uniform-discrete} or~\eqref{eq:tau-basis}.
\end{definition}

We will see that the $\tau$-uniformity condition ensures that the system follows an intrinsic effective Schr\"odinger-like equation with respect to $\tau$, but I will not consider it mandatory.

This proposal achieves what the Page-Wootters proposal achieves, but without introducing a clock external to the world or separating a clock-subsystem out of the world, and whose Hamiltonian $\obs{H}_c$ is unbounded, which would also make the total Hamiltonian $\obs{H}=\obs{H}_c+\obs{H}_r$ unbounded.
By contrast, the present proposal is based on identifying the time observable based on the records, which is the logical way~\cite{Stoica2026PageWoottersRelationalQuantumTimeIsReversed}.

In the proposal from this article, there is no external clock, the time observable being intrinsic and collective to the system, and compatible with a Hamiltonian that is bounded from below.

%------------------------------------------------------------%
\subsection{The intrinsic Schr\"odinger equation}
\label{s:intrinsic-schrod}

For an intrinsic $\tau$-uniform time observable, for any $\ket{\psi}\in\hilbert$, the vector $\obs{H}\ket{\psi}=\int_{\R} c(\tau)\ket{\psi(\tau)}\dd\tau$ identifies a minimal closed subspace $\hilbert_{\ket{\psi}}$ of $\hilbert$ spanned by generalized eigenvectors $\ket{\psi(\tau)}$ of $\obs{T}$. Due to Implication~\ref{implication:UW}, this subspace contains a generalized eigenvector of $\obs{T}$ for each $\tau\in\R$.
On $\hilbert_{\ket{\psi}}$, we define the translation operator
\begin{equation}
\label{eq:clock-H-psi}
\widetilde{\obs{H}}_{\ket{\psi}}=-i\hbar\pdv{\tau}.
\end{equation}
This extends to the full Hilbert subspace to a self-adjoint operator generating translations on the $\tau$-basis of $\widetilde{\hilbert}$,
\begin{equation}
\label{eq:clock-H}
\widetilde{\obs{H}}=-i\hbar\pdv{\tau}
\end{equation}
so that
\begin{equation}
\label{eq:clock-translation}
\ket{\psi(\tau')}=e^{-\frac{i}{\hbar}\widetilde{\obs{H}}(\tau'-\tau)}\ket{\psi(\tau)},
\end{equation}
so the following effective Schr\"odinger equation holds
\begin{equation}
\label{eq:tau-schrod}
\begin{aligned}
i\hbar\dv{\tau}\ket{\psi(\tau)}=\widetilde{\obs{H}}\ket{\psi(\tau)}.
\end{aligned}
\end{equation}

This intrinsic Schr\"odinger equation has Schr\"odinger time $\tau$, and it describes the $\tau$-evolution as seen from the intrinsic perspective.
The intrinsic Hamiltonian $\widetilde{\hilbert}$ is unbounded, since its spectrum is $\R$.

We notice that
\begin{equation}
\begin{aligned}
[\obs{T},\widetilde{\obs{H}}]\ket{\psi(\tau)}
=&\obs{T}\widetilde{\obs{H}}\ket{\psi(\tau)}-\widetilde{\obs{H}}\(\tau\ket{\psi(\tau)}\) \\
=&-i\hbar\obs{T}\pdv{\tau}\ket{\psi(\tau)}+i\hbar\ket{\psi(\tau)} \\
&+i\hbar\obs{T}\pdv{\tau}\ket{\psi(\tau)} \\
=&i\hbar\ket{\psi(\tau)}, \\
\end{aligned}
\end{equation}
so $\obs{T}$ and $\widetilde{\obs{H}}$ satisfy the \emph{canonical commutation relation} on the suitable domain,
\begin{equation}
\label{eq:clock-time-operator}
[\obs{T},\widetilde{\obs{H}}]=i\hbar.
\end{equation}

The evolution is unitary, due to the condition that the time observable is uniform~\eqref{eq:tau-uniform-continuum}.
The intrinsic history of the system,
\begin{equation}
\label{eq:tau-history}
\tau\mapsto\Psi(\tau,\M,\q),
\end{equation}
can be expressed as a solution of a Schr\"odinger equation like~\eqref{eq:clock-H} in which the Hamiltonian $\widetilde{\obs{H}}$ is expressed as $\widetilde{\obs{H}}(\M,\q)$, in terms of the configuration space variables $\M$ and $\q$,
\begin{equation}
\label{eq:tau-hamiltonian-obs}
\begin{aligned}
\widetilde{\obs{H}}(\M,\q)\ket{\psi(\tau)}
:=&\matrixel{\tau,\M,\q}{\widetilde{\obs{H}}}{\psi(\tau)} \\
=&-i\hbar\matrixel{\tau,\M,\q}{\pdv{\tau}}{\psi(\tau)} \\
\end{aligned}
\end{equation}

Then, any intrinsic $\tau$-history of the system as in~\eqref{eq:tau-history} can be described as a solution to a Schr\"odinger equation whose time parameter is $\tau$. Its Hamiltonian $\widetilde{\obs{H}}(\M,\q)$ is, in general, different from the original Hamiltonian $\obs{H}(\M,\q)$ from equation~\eqref{eq:time-evol}, for example its spectrum is unbounded.

However, if we are serious about getting rid of unobservable quantities, we should also acknowledge that we have no way to access the amplitude at $\tau$, because all we can access are the intrinsic records. Thus, even the ``intrinsic Schr\"odinger equation'' is a misnomer, since we would not be able to distinguish it for similar equations obtained for non-uniform time observables.
The ``true'' intrinsic Schr\"odinger equation should also be consistent with the relative changes as registered in the records.

%------------------------------------------------------------%
\section{Conclusions}
\label{s:conclusions}

Since we can distinguish macroscopic states at different times, it follows that even if the state of the world would be a superposition of states at different times, we would not know it.
In fact, states with definite time are extremely nongeneric, and for any definite measure, they have vanishing probability.
Moreover, if for some value of the Schr\"odinger parameter $t$, the state of the world has a definite intrinsic time, if the Hamiltonian is bounded from below, the state would immediately reach states at all intrinsic times from the past or from the future.
Therefore, the problem is unavoidable and needs attention.

We have seen that accepting that the intrinsic time that we measure relies on macroscopic pointer states, and that observers never observe superpositions of distinct macroscopic pointer states, allows us to define good time observables, even though they do not correlate monotonically with the Schr\"odinger parameter $t$.
It then becomes justified to demote the Schr\"odinger parameter $t$ from the role of time and accept instead intrinsic time observables.
This leads to the Wheeler-DeWitt equation, even for quantum theories that do not explicitly include gravity.

We have also seen that, if the time observable $\obs{T}$ is such that the state's amplitude is spread uniformly along $\tau$, an effective intrinsic Schr\"odinger equation emerges.

An open question is how do we recover the Schr\"odinger equation from the intrinsic perspective. I expect that we can infer the Hamiltonian of a subsystem from the intrinsic perspective of its environment, rather than the universal Schr\"odinger equation, which cannot be reached from within its own solutions.

In a follow-up article~\cite{Stoica2026EmergentCosmologyAndGravityFromQuantumTime} I show that, if the uniformity condition is satisfied, it induces a curvature in the time direction, allowing the emergence of Friedmann-Lema\^itre-Robertson-Walker cosmological models, including the $\Lambda$CDM model.
If four independent quantum time observables exist, they can parametrize spacetime and induce spacetime curvature~\cite{Stoica2026EmergentCosmologyAndGravityFromQuantumTime}.

%------------------------------------------------------------%

\end{document}